\documentclass[12pt]{article}

\usepackage{lineno,hyperref}
\usepackage[utf8]{inputenc}
\usepackage{fontawesome}
\usepackage{caption}
\usepackage{subcaption}
\PassOptionsToPackage{dvipsnames,svgnames,table}{xcolor}
\usepackage{xcolor}
\usepackage[affil-it]{authblk}

\usepackage{amsmath}
\usepackage{amsfonts}
\usepackage{amssymb}
\usepackage{mathtools}
\usepackage{amsthm}
\usepackage{nicefrac}
\usepackage{cases}
\usepackage{mathrsfs}

\usepackage{tikz}
\usetikzlibrary{calc, math, arrows, arrows.meta, decorations.pathreplacing}
\usetikzlibrary{matrix}

\newtheorem{theorem}{Theorem}[section]

\newtheorem{remark}[theorem]{Remark}
\newtheorem{lemma}[theorem]{Lemma}

\newtheorem{observation}[theorem]{Observation}
\newtheorem{example}[theorem]{Example}

\newcommand*{\defeq}{\mathrel{\vcenter{\baselineskip0.5ex \lineskiplimit0pt
			\hbox{\scriptsize.}\hbox{\scriptsize.}}}%
	=}
\DeclareMathOperator*{\Z}{\mathbb{Z}}
\DeclareMathOperator*{\N}{\mathbb{N}}
\DeclareMathOperator*{\NP}{\mathsf{NP}}

\DeclareMathOperator*{\Paths}{\mathcal{P}}

\newcommand{\ceil}[1]{\lceil #1 \rceil}

\newcommand{\opts}[2]{\ensuremath{\operatorname{OPT}(#1,\sum,#2)}}

\newcommand{\card}[1]{\left|#1\right|}

\newcommand{\Ol}{\mathcal{O}}
\usepackage{setspace}
\tikzset{zeilenabstand/.style={%
		execute at begin node=\begin{spacing}{#1}\hspace{0pt},%
			execute at end node=\vspace{-\ht\strutbox}\end{spacing}}}

\definecolor{myGray}{gray}{0.82}
\definecolor{myGrayLight}{gray}{0.91}
\definecolor{myGrayDark}{gray}{0.5}
\newcommand{\problem}[3]{
	\begin{center}
		\begin{tabular}{p{0.13\linewidth} p{0.8\linewidth}}
			\multicolumn{2}{ l}{#1}\\ 
			\textsc{Instance:} & 	#2 \\
			\textsc{Task:} & 	#3 \\ 
			
		\end{tabular}
	\end{center}
}
\newcommand{\decision}[3]{
	\begin{center}
		\begin{tabular}{p{0.13\linewidth} p{0.8\linewidth}}
			\multicolumn{2}{ l}{#1}\\ 
			\textsc{Instance:} & 	#2 \\
			\textsc{Question:} & 	#3 \\ 
			
		\end{tabular}
	\end{center}
}

\title{\(p\)-median location interdiction on trees}
\author[a]{L. Leiß\footnote{Corresponding author, Email address: \texttt{leiss@mathematik.uni-kl.de} (L. Leiß)}}
\author[b]{T. Heller}
\author[c]{L. Schäfer}
\author[d]{M. Streicher}
\author[a]{S. Ruzika}

\affil[a]{\footnotesize Department of Mathematics, RPTU Kaiserslautern-Landau, 67663 Kaiserslautern, Germany}
\affil[b]{\footnotesize Department of Optimization, Fraunhofer Institute for Industrial Mathematics ITWM, 67663 Kaiserslautern, Germany}
\affil[c]{\footnotesize Comma Soft AG, 53229 Bonn, Germany}
\affil[d]{\footnotesize PTV Group, 76131 Karlsruhe, Germany}
\date{}

\begin{document}
\maketitle
\vspace{-3em}
\begin{abstract}
In \(p\)-median location interdiction the aim is to find a subset of edges in a graph, such that the objective value of the \(p\)-median problem in the same graph without the selected edges is as large as possible.

We prove that this problem is \(\NP\)-hard even on acyclic graphs. Restricting the problem to trees with unit lengths on the edges, unit interdiction costs, and a single edge interdiction, we provide an algorithm which solves the problem in polynomial time. Furthermore, we investigate path graphs with unit and arbitrary lengths. For the former case, we present an algorithm, where multiple edges can get interdicted. Furthermore, for the latter case, we present a method to compute an optimal solution for one interdiction step which can also be extended to multiple interdicted edges. 
\end{abstract}


\section{Introduction}\label{sec:intro}
Location planning is a field of mathematical research which crosses our daily life more often, than we might think at first sight. The root of modern location planning goes back to Pierre de Fermat and aims at finding the point, which minimizes the sum of the Euclidean distances of three given points to the new location (cf. \cite{drezner2001facility}). A popular, more applied version of this problem is the identification of a new location for a supplier of materials for further industrial processing which has been stated in \cite{weber1922Standort} and is called Weber problem (cf. \cite{drezner2001facility}). In a more general approach, a new location is to be found which minimizes the sum of all - possibly weighted - distances from all given locations to the new one. This problem is called median location problem (cf. \cite{laporte2015location}). The underlying structure, on which location problems can be analyzed, may vary. Mainly, we distinguish between planar location problems and network location problems. In this article, we consider the network case only. A further alteration of the main problem comes with the number of new locations to be computed. We refer to the problem of placing \(p\) new facilities as the \(p\)-median location problem. The decision version of the \(p\)-median location problem is known to be \(\NP\)-complete for variable \(p\) on general networks, which can be shown by a reduction from the \textit{dominating set problem} (\cite{kariv1979algorithmicmedian}). For this case, there are several heuristics known. A good overview on this topic can for example be found in \cite{mladenovic2007p}. For a general graph with unit length values on the edges and variable \(p\), the problem still remains \(\NP\)-complete (\cite{NP}). In contrast, for fixed \(p\), the decision problem is solvable in polynomial time on general graphs by enumeration (cf. \cite{NP}). Due to the hardness of the general case, research focused on particular graph structures. For \(G\) being a tree, the authors in \cite{kariv1979algorithmicmedian} present an \(\Ol(n^2p^2)\)-algorithm to solve the problem. A dynamic programming approach based on this result was later proposed in \cite{tamir1996pn2}, which runs in \(\Ol(pn^2)\) and got improved by \cite{benkoczi2005new} to an algorithm which runs in \(\Ol(n\log^{p+2}n)\). The complexity for path graphs is shown to be \(\Ol(pn)\) in \cite{hassin1991improved}.  A linear time algorithm can be applied for the \(1\)-median location problem on trees (\cite{goldman1971optimal}). 

A good overview of location planning can for example be found in \cite{drezner2001facility}, \cite{hale2003location} or \cite{laporte2015location}. An overview of the solution methods for the \(p\)-median location problem in particular can be found in \cite{reese2006solution}.
\par
Interdiction problems pursue the question of how a system can be interrupted in the worst possible way in terms of the original objective function. The interruption itself can be caused by different acts, such as modification of the edge lengths or deletion of entire edges as well as the vertices. 
\par

Interdiction problems have gained increasingly more attention lately (cf.~\cite{smith2020}). There are several different applications, that motivate research in this field. The interdiction of a network can either have a desirable outcome, such as in narrowing the spread of a disease (cf. \cite{assimakopoulos1987network}), interdicting smuggling routes (cf. \cite{morton2007models}) or -- as for example done in \cite{wood1993deterministic} -- the aim of (armed) forces to reduce the amount of drugs and chemicals transported illegally via road or waterways -- possibly with limited resources. Also -- on the contrary --  attacks on networks can be interpreted via interdiction steps. In this case, the analysis of these problems might allow the determination of valuable edges or locations of the original network. \par 
Two main optimization problems, which have been studied in the context of interdiction, are the \textit{shortest path problem} as well as the \textit{maximum flow problem}. 
Notable research on the first topic has been done in \cite{ball1989finding}, \cite{bar1998complexity}, \cite{corley1982most} or \cite{israeli2002shortest},  for instance. Results on the latter problem might for example be found in \cite{ghare1971optimal}, \cite{ratliff1975finding}, \cite{schafer2020bicriterion}, \cite{wollmer1964} or \cite{wood1993deterministic}. In \cite{smith2020}, one can find a recent overview of the literature on interdiction problems.

Research concerning the combination of location and interdiction problems on the other hand is quite scarce. In this context, we mention the \(r\)-interdiction median problem, which is defined as follows. Given a supply-system with \(p\) existing locations, the interdictor wishes to find the subset of \(r\) locations, which, when removed, yields the highest weighted distance with respect to the median location function (cf. \cite{church2004identifying}). There are a few extensions to this problem, such as the possibility to fortify a fixed number of the existing facilities, which in consequence cannot be interdicted by the attacker anymore (cf. \cite{liberatore2011analysis}). In \cite{aksen2010budget}, the authors alter the concept of fortifying a specified number of locations, but rather introduce a restricted budget for fortification. A further topic, gaining more interest, is the $p$-hub interdiction problem, which is for example considered in \cite{ullmert2020p}. Still, most of these approaches have in common, that the interdictor intervenes the existing locations and not the underlying network. An exception to this concept can be found in \cite{frohlich2021hardness}. Here, the authors combine the \textit{covering problem} with an edge interdiction problem.  In \cite{frohlich2022interdicting}, the authors present a polynomial time algorithm for the interdiction problem on trees, where an upfront chosen set of facilities is given and the interdictor wishes to worsen the reachability within the tree. In \cite{frohlich2021facility}, the authors combine the \textit{median location problem} with edge interdiction. To the best of our knowledge, this is the only work on the $p$-median interdiction problem with edge interdiction. There, they consider the problem for different orders of action of the locator and the interdictor. For the case, that the interdictor acts before the locator, they prove this problem to be \(\Sigma_2^p\)-complete in the general case. In the same work, a bilevel mixed-integer formulation is presented for the problem. This result motivates the analysis of the problem for restricted cases.

\paragraph{Our contribution}
Given the complexity analysis in \cite{frohlich2021facility}, it remains open, if efficient solution procedures can be found if the general problem is restricted. Coming from the original \(p\)-median problem, which is solvable in polynomial time on trees, an obvious variant of the corresponding interdiction problem is to restrict the problem to trees (or even simpler structures) as well. For such cases, complexity results and solution methods are not yet available. 

In this article, we aim at closing this research gap.  We consider the median location problem in combination with an interdictor who can delete edges in a given network and wishes to maximize the objective function value of the locator. We assume, that the interdiction step is executed before the locator places their median facility.  Due to sigma-2-p hardness for the general case, we consider the problem on particular graph structures. We analyze the complexity of the general problem on trees and present an algorithm to solve the problem exactly for some variant. Furthermore, we present an algorithm for graphs with path structure for both cases of unit and arbitrary lengths on the edges. 

\paragraph{Outline}
The remainder of this article is structured as follows. In Section~\ref{sec:preliminaries}, we give a short overview of the concepts needed for the article and define the investigated problem. Section~\ref{sec:complexity} deals with the complexity of the median location interdiction problem. Section~\ref{sec.path} studies the strategy for solving the interdiction median problem on paths with unit length values on the edges. Furthermore, we present a strategy for paths with arbitrary lengths. The next Section~\ref{sec.tree} focuses on a tree with unit length values. We state an algorithm, which solves the problem exactly in polynomial time. Section~\ref{sec:conclusion} then summarizes the paper and proposes further directions of research.

\section{Preliminaries and problem formulation}\label{sec:preliminaries}
Let \(G=(V(G),E(G))\) be an undirected graph with vertex set \(V(G)=\{v_1,\dots,v_n\}\) and edge set \(E(G)=\{e_1,\dots,e_m\}\), where \(n\defeq|V(G)|\) and \(m\defeq|E(G)|\). If the underlying graph is known by the context, we refer to \(V(G)\) as \(V\) and to \(E(G)\) as \(E\). Also, for better readability, instead of \(\card{V(G)}\), we may write \(\card{G}\). For a given vertex \(v\), we denote the number of incident edges, i.e. its degree, by \(\deg(v)\). Further, we assign a length value to each edge~\(e\in E\), i.e., \(\ell\colon E\rightarrow\Z_+\). Let \(\Paths_{uv}\) be the set of all paths \(P\) connecting vertices \(u, v\in V\). Then, the length of a shortest path between the vertices~\(u\) and \(v\) is denoted by \(d(u,v)\), i.e., \(d(u,v)\defeq\min\limits_{P\in\Paths_{uv}}\ell(P)\). If the graph on which the distance is measured is not clear from context we also write $d_G(u,v)$.

If \(G\) is a tree, let some vertex \(r\in V\) be the root of the breadth-first-graph of \(G\) (\cite{leiserson1994introduction}). We denote by \(G_{v}\) the subtree of \(G\), which is rooted in vertex \(v\) and contains all descendants of vertex \(v\) in the breadth-first-graph of \(G\) with root \(r\) as well as their incident edges.\\

As stated, given an instance of the median problem, one aims at placing one (or more) new location(s), which minimize(s) the sum of the shortest path lengths of all existing locations to their nearest facility. In this article, we consider the case, that the set of existing locations is the vertex set. A chosen set \(X\) of locations therefore has the objective value:
\[
f(X) \defeq \sum_{v\in V} d(v,X), \text{ with } d(v,X) \defeq \min_{x\in X}d(v, x).
\]
Based on this objective, we state the \(p\)-median location problem as follows.
\problem{\(p\)-median location problem \((p,\sum, G)\)}{Undirected graph \(G=(V,E)\), edge lengths \(\ell\colon E\rightarrow\Z_+\), and number of locations \(p\in\Z_+\).}{Find a set \(X\subseteq V \) of \(p\) new locations such that the objective function of the \(p\)-median location problem is minimal, i.e. minimize \[\sum_{v\in V} d(v_i,X).\]}
The optimal solution is denoted by \(X^*\), while the optimal objective function value is denoted by \(\opts{p}{G}\).

Network interdiction problems involve an additional opposing force, called the interdictor. Said interdictor wishes to worsen the objective function value of the locator. In this article, the interdictor is constrained by an interdiction budget \(B\in\Z_+\), while each edge \(e\in E\) is associated with an interdiction cost \(b(e)\in\Z_+\), i.e., \(b\colon E\rightarrow\Z_+\).
Consequently, the set of all feasible interdiction strategies, denoted by \(\Gamma\), can be expressed as follows:
\begin{equation*}
	\Gamma \defeq \left\{\gamma = (\gamma_e)_{e \in E} \in \{0,1\}^m \mid \sum\limits_{e \in E} b(e)\cdot\gamma_e \leq B\right\},	
\end{equation*}
where \(\gamma_e\) equals one, if edge \(e\) is interdicted or zero, if not. The set of optimal interdiction strategies is denoted by \(\Gamma^*= \left\{\gamma^*\in\Gamma\mid\gamma^*\text{ optimal}\right\}\). 
In what follows, each interdiction strategy \(\gamma \in \Gamma\) induces an undirected graph \(G(\gamma)\defeq(V', E')\) with \(V'=V\) and \(E'=E\setminus E(\gamma)\), where \(E(\gamma)\defeq\{e\in E\mid \gamma_e=1\}\). In this article, the locator places their facility on \(G(\gamma)\), i.e. after the interdiction step.
Based on this, we define the decision version of the \(p\)-median location interdiction problem as follows.

\decision{Decision version of the \(p\)-median location interdiction problem}{Undirected graph \(G=(V,E)\), edge lengths \(\ell\colon E\rightarrow\Z_+\), interdiction costs \(b\colon E\rightarrow\Z_+\), interdiction budget \(B\in\Z_+\), number of locations \(p\in\Z_+\), and decision parameter \(K\in\Z_+\).}{Does there exist an interdiction strategy \(\gamma\in\Gamma\) such that \[\min_{X\subseteq V, \card{X}=p} \sum_{v\in V} d_{G(\gamma)}(v,X)\geq K \ \text{?}\]}
In the optimization version of the stated problem, we aim to find the maximum \(K\) for which the decision version is a \textsc{yes}-instance.

\section{Complexity results}\label{sec:complexity}
It is well known that the $p$-median location problem is $\NP$-complete, cf.~\cite{kariv1979algorithmicmedian}. Therefore it would be surprising for the corresponding interdiction problem to be polynomial time solvable. In fact Fröhlich \cite{frohlich2021facility} proved that the decision version of the $p$-median location interdiction problem is $\Sigma_2^p$-complete. However, there are several restrictions of the $p$-median location problem which are proven to lead to polynomial time solvability, as mentioned in the introduction. Among the  restricted versions, that are polynomial time solvable is the $p$-median location problem on trees, cf.~\cite{tamir1996pn2}. In this section we prove that adding the interdiction layer to the problem makes the problem significantly harder: The $p$-median location interdiction problem is $\NP$-complete even on trees.

For this, we consider the \textsl{knapsack problem with bounded profit ratio of 2 (K-BPR2)}. An instance is given by a set \(M\) of \(m\) items with associated weights \(w_i\geq 0\) and profits \(p_i\geq 0 \ \text{ for all } i\in\left\{1,\dots,m\right\}\). For the profits it holds \(\frac{p_i}{p_j}\leq 2 \text{ for all } p_i,p_j\) with \(i\neq j\). We now show that this problem is $\NP$-complete. 
\begin{lemma}
	The \textsl{knapsack problem with bounded profit ratio of 2} is \(\NP\)-complete.
\end{lemma}

\begin{proof}
	We reduce the \textsl{equal partition problem} to \textsl{K-BPR2}. Let an instance of the \textsl{equal partition problem}, which is known to be \(\NP\)-complete (cf.~\cite{NP}), be given with a set \(I=\left\{\tilde{w}_1,\dots,\tilde{w}_n\right\}\) such that \(\sum_{i=1}^{n}\tilde{w}_i=B\in 2\Z, \tilde{w}\geq 0\). For the \textsl{equal partition problem} we ask for a partitioning of the elements of \(I\) into two subsets \(I_1,I_2\) such that \(\sum_{i\in I_1}\tilde{w}_i=\sum_{i\in I_2}\tilde{w}_i=\frac{B}{2}\) and \(\card{I_1}=\card{I_2}=\frac{n}{2}\). Define \(p_i=w_i\defeq \tilde{w}_i+B+1 \ \text{ for all } i\). According to this definition, the smallest \(p_i\) is at least \(B+1\) and the biggest \(p_i\) is not greater than \(2B+1\) such that the ratio \(\frac{p_i}{p_j}\) of all pairs \(p_i,p_j, i\neq j\) is bounded by \(2\).\\
	Suppose we are given a solution \(I_1\) to the \textsl{equal partition problem}, we ask for a solution to the \textsl{knapsack problem} with profit \[P\geq \frac{B}{2}+\frac{n}{2}(B+1)\] and total weight \[W\leq \frac{B}{2}+\frac{n}{2}(B+1).\]We state, that the selection \(I_1\) is such a solution. For the profit, it is \(P=\sum_{i\in I_1}p_i=\sum_{i\in I_1}(\tilde{w}_i+B+1)=\frac{B}{2}+\frac{n}{2}(B+1)\). Also, \(W=\sum_{i\in W_1}w_i=\sum_{i\in I_1}(\tilde{w}_i+B+1)=\frac{B}{2}+\frac{n}{2}(B+1)\).\\
	On the other hand, given a solution \(J\in I\) to the \textsl{knapsack problem} for which \(P\geq \frac{B}{2}+\frac{n}{2}(B+1)\) and \(W\leq \frac{B}{2}+\frac{n}{2}(B+1)\) we show that \(J\) together with \(I\setminus J\) is a solution to the \textsl{equal partition problem}. To prove this statement, we need to show that \(\card{J}=\frac{n}{2}\). Suppose that \(\card{J}<\frac{n}{2}\). Then for the profit, it is \(\sum_{j\in J}p_j=\sum_{j\in J}(\tilde{w}_j+B+1)\leq B + (\frac{n}{2}-1)(B+1)=\frac{n}{2}(B+1)-1<P \) which is a contradiction to the solution of the \textsl{knapsack problem}. Analogously, suppose that \(\card{J}>\frac{n}{2}\). Then for the weight it is \(\sum_{j\in J}w_j=\sum_{j\in J}(\tilde{w}_j+B+1)\geq B + (\frac{n}{2}+1)(B+1)=2B+1+\frac{n}{2}(B+1)>W\). Therefore, we have \(\card{J}=\frac{n}{2}\). Finally, calculate \[\sum_{j\in J}p_j=\sum_{j\in J}\tilde{w}_j+\frac{n}{2}(B+1)\geq \frac{B}{2}+\frac{n}{2}(B+1)\] and also \[\sum_{j\in J}w_i=\sum_{j\in J}\tilde{w}_j+\frac{n}{2}(B+1)\leq\frac{B}{2}+\frac{n}{2}(B+1).\]Therefore, we get that \(\sum_{j\in J}\tilde{w}_j=\frac{B}{2}\).
	
	Clearly, given a solution to the \textsl{K-BPR2}, we can verify this solution in polynomial time. Thus, the \textsl{K-BPR2} is $\NP$-complete. 
\end{proof}

\begin{theorem}~\label{thm:compl}
	The \(p\)-median location interdiction problem, where the underlying graph is given by a tree, is \(\NP\)-complete.
\end{theorem}

\begin{proof}
	We show this by reducing \textsl{K-BPR2} to the location interdiction problem. Let an instance of \textsl{K-BPR2} be given by a set~\(M\) of \(m\) items with associated weights \(w_i\) and profits~\(p_i \ \text{ for all } i\in\left\{1,\dots,m\right\}\). For the profit it holds for all \(p_i,p_j\) with~\(i\neq j\) that \(\frac{p_i}{p_j}\leq 2\). Furthermore, let \(W, P\in\Z_+\) be two integers. Consider a tree with lengths~\(\ell_i\) and interdiction costs~\(b_i\) for every edge~\(e_i\) as depicted in Figure~\ref{fig:nptree}. Note, that the construction provides \(m\) paths of four vertices emerging from vertex~\(v_x\) and one path of two vertices. Let \(B=W\) be the interdiction budget. For every selection of interdicted edges, one optimal strategy to choose \(m+1\) centers is depicted in Figure~\ref{fig:nptreesol}.
	
	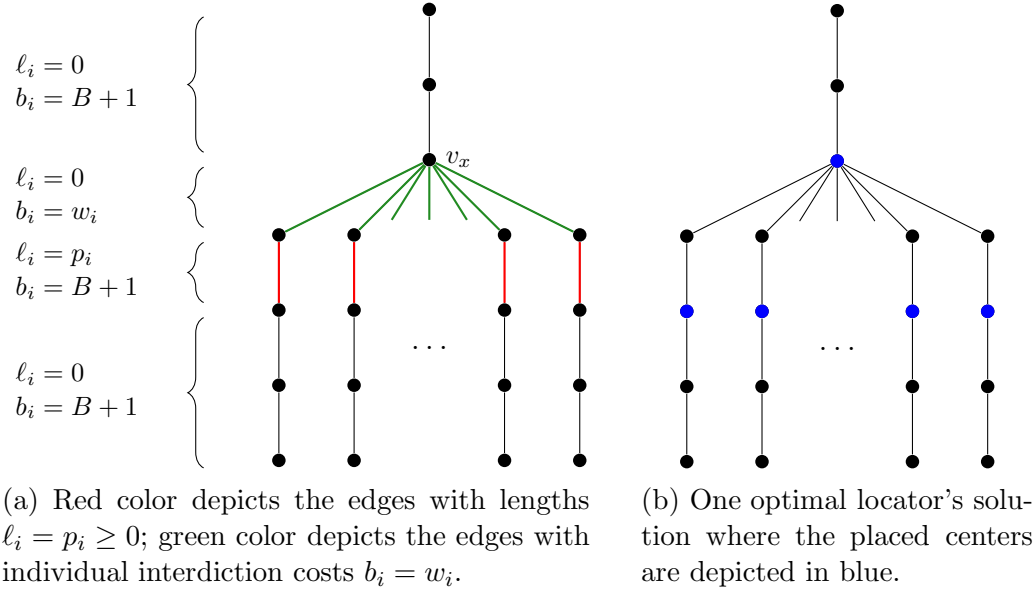
\begin{figure}
		\centering
		\subcaptionbox{Red color depicts the edges with lengths \(\ell_i=p_i\geq 0\); green color depicts the edges with individual interdiction costs \(b_i=w_i\).  \label{fig:nptree}}%
		[.57\textwidth]{\begin{tikzpicture}
				\begin{scope}
					[every node/.style={shape=circle,fill=black, inner sep=1.8}]
					\foreach \x in {-1, 0}{
						\foreach \y [evaluate=\y as \z using int( 4*(\x+1)+\y+2 )] in {-1, 0, 1, 2}{
							\node (\z) at (\x,\y){};
						}
					}
					\foreach \x in {2, 3}{
						\foreach \y [evaluate=\y as \z using int( 4*(\x)+\y+2 )] in {-1, 0, 1, 2}{
							\node (\z) at (\x,\y){};
						}
					}
					\node (17) [label={[label distance=0cm]0:\footnotesize\(v_x\)}] at (1,3){};
					\node (18) at (1,4){};â
					\node (19) at (1,5){};
					\draw (1) -- (2) -- (3) -- (4) -- (17) -- (8) -- (7) -- (6) -- (5);
					\draw (9) -- (10) -- (11) -- (12) -- (17) -- (16) -- (15) -- (14) -- (13);
					\draw (17) -- (18) -- (19);
					\draw[thick, ForestGreen] (17) -- (.5,2.2);
					\draw[thick, ForestGreen] (17) -- (1,2.2);
					\draw[thick, ForestGreen] (17) -- (1.5,2.2);
					\draw[thick, red] (3) -- (4);
					\draw[thick, red] (7) -- (8);
					\draw[thick, red] (11) -- (12);
					\draw[thick, red] (15) -- (16);
					\draw[thick, ForestGreen] (4) -- (17) -- (8);
					\draw[thick, ForestGreen] (12) -- (17) -- (16);
				\end{scope}
				\node (d) at (1,.5) {\(\dots\)};
				\begin{scope}[every node/.style={text width=2cm,zeilenabstand=.5,label distance=-2pt, midway, xshift=-1.5cm}]
					\draw [decorate,decoration={brace,amplitude=6pt}] (-2,3.1) -- (-2,4.9) node {\footnotesize \(\begin{aligned}  \ell_i&=0 \\ b_i&=B+1 \end{aligned}\)};
					\draw [decorate,decoration={brace,amplitude=6pt}] (-2,2.1) -- (-2,2.9) node {\footnotesize\(\begin{aligned}  \ell_i&=0 \\ b_i&=w_i \end{aligned}\)};
					\draw [decorate,decoration={brace,amplitude=6pt}] (-2,1.1) -- (-2,1.9) node {\footnotesize\(\begin{aligned} \ell_i&=p_i\\ b_i&=B+1 \end{aligned}\)};
					\draw [decorate,decoration={brace,amplitude=6pt}] (-2,-1.1) -- (-2,0.9) node {\footnotesize\(\begin{aligned}  \ell_i&=0 \\ b_i&=B+1 \end{aligned}\)};
				\end{scope}
		\end{tikzpicture}}
		\hfill
		\subcaptionbox{One optimal locator's solution where the placed centers are depicted in blue. \label{fig:nptreesol}}
		[.38\textwidth]{
			\begin{tikzpicture}
				\begin{scope}
					[every node/.style={shape=circle,fill=black, inner sep=1.8}]
					\foreach \x in {-1, 0}{
						\foreach \y [evaluate=\y as \z using int( 4*(\x+1)+\y+2 )] in {-1, 0, 1, 2}{
							\node (\z) at (\x,\y){};
						}
					}
					\foreach \x in {2, 3}{
						\foreach \y [evaluate=\y as \z using int( 4*(\x)+\y+2 )] in {-1, 0, 1, 2}{
							\node (\z) at (\x,\y){};
						}
					}
					\node (17) at (1,3){};
					\node (18) at (1,4){};
					\node (19) at (1,5){};
					\draw (1) -- (2) -- (3) -- (4) -- (17) -- (8) -- (7) -- (6) -- (5);
					\draw (9) -- (10) -- (11) -- (12) -- (17) -- (16) -- (15) -- (14) -- (13);
					\draw (17) -- (18) -- (19);
					\node[thick, blue] at (3) {};
					\node[thick, blue] at (7) {};
					\node[thick, blue] at (11) {};
					\node[thick, blue] at (15) {};
					\node[thick, blue] at (17) {};
				\end{scope}
				\node (d) at (1,.5) {\(\dots\)};
				\draw (17) -- (.5,2.2);
				\draw (17) -- (1,2.2);
				\draw (17) -- (1.5,2.2);
		\end{tikzpicture}}
		\caption{The tree graphs used in the proof of Theorem~\ref{thm:compl}. \label{fig:nptreetree}}
	\end{figure}
	
	To prove this fact, we need to show that every single path emerging from vertex \(v_x\) must have one center, and furthermore, that the center in the paths with 4 vertices must be below the edges where \(\ell_i=p_i\).\\
	In case that every outgoing edge of \(v_x\) which is part of the \(m\) paths of 4 vertices (depicted in green in Figure~\ref{fig:nptree}) is interdicted, the stated solution is clearly optimal. Now consider the case where w.l.o.g. the leftmost interdictable edge is interdicted and let the second outgoing edge of \(v_x\) be non interdicted. The notation used  can be found in Figure~\ref{fig:npredtree}.
	
	\begin{figure}
		\centering
		\begin{tikzpicture}
			\begin{scope}
				[every node/.style={shape=circle,fill=black, inner sep=1.8}]
				\foreach \x/\y/\z in {0/0/1, 0/1/2, 0/2/3, 0/3/4, 1/4/5, 2/3/6, 2/2/7, 2/1/8, 2/0/9, 1/5/10, 1/6/11} \node (\z) at (\x,\y) {};
				\draw (1) -- (4); 
				\draw (5) -- (6)-- (9);
				\draw (5) -- (11);
				\foreach \r in {3, 5, 7} \node[thick, blue] at (\r){};
			\end{scope}
			\node [label={[label distance=0cm]0:\footnotesize\(v_x\)}] at (5){};
			\node [label={[label distance=0cm]0:\footnotesize\(u_2\)}] at (7){};
			\node [label={[label distance=0cm]0:\footnotesize\(v_2\)}] at (6){};
			\node [label={[label distance=0cm]180:\footnotesize\(u_1\)}] at (3){};
			\node [label={[label distance=0cm]180:\footnotesize\(v_1\)}] at (4){};
			\node [label={[label distance=0cm]180:\scriptsize\(\ell_1\)}] at (0,2.5) {};
			\node [label={[label distance=0cm]0:\scriptsize\(\ell_2\)}] at (2,2.5) {};
			\foreach \j/\k in {1/4.5,1/5.5,2/0.5,2/1.5} \node [label={[label distance=-.25cm]0:\scriptsize\(0\)}] at (\j,\k) {};
			\foreach \j/\k in {0/0.5,0/1.5} \node [label={[label distance=-.25cm]180:\scriptsize\(0\)}] at (\j,\k) {};

		\end{tikzpicture}
		\caption{Excerpt of the tree from Figure~\ref{fig:nptreetree}.\label{fig:npredtree}}
	\end{figure}
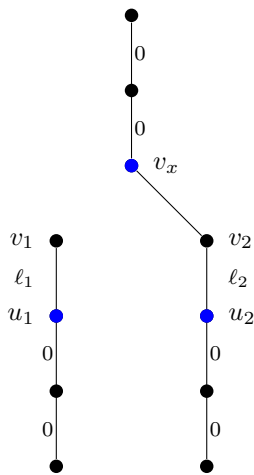
	It is clear, that every path with interdicted starting edge (the edge outgoing of \(v_x\)) needs at least one center for feasibility. Also, this center has to be placed below the edge with length greater than zero -- in the considered case \(u_1\) or a vertex below. Now, the only possibility for the locator to change the objective function value is to spare the center of the non interdicted paths -- \(u_2\) in the example --  and instead place it at \(v_1\). Then, the length \(\ell_1\) does not appear in the calculation of the objective function value. But instead, all vertices in the second path then need to be covered via center \(v_x\). That means, the edge with length \(\ell_2\) needs to be crossed 3 times. Even if we assume that the lengths have the maximal possible factor \(\ell_1=2\cdot\ell_2\), it would still be better for the locator to place the center at vertex \(u_2\). The same explanation holds for the shifting of the center in \(v_x\) to a vertex in an interdicted path. In this case, the vertices above \(v_x\) need to be covered by another center below which is -- by the same estimation as before -- worse than the provided solution of Figure~\ref{fig:nptreesol}. Note that this solution is not unique. In fact -- as stated before -- the center in the \(m\) paths of 4 vertices can be placed at any vertex below the edges where \(\ell_i=p_i\). Also, the center at \(v_x\) can be shifted up to two vertices up.\\
	Assume we are given a solution of the \textsl{knapsack problem} such that the selection \(I\in M\) of items fulfills \(\sum_{i\in I}w_i\leq W\) and \(\sum_{i\in I}p_i\geq P\). Now consider the optimal solution of the locator as stated above after interdiction of the edges  \(e_i, i\in I\).  For the interdiction costs, it is \(\sum_{i\in I}b_i=\sum_{i\in I}w_i\leq W\). Furthermore, the corresponding objective function value calculates as \(\sum_{i\in I}\ell_i=\sum_{i\in I}p_i\geq P\).\\

	On the other hand, if we are given a selection of interdicted edges \(J\in E_M\) such that \(\sum_{j\in J}b_j\leq B\) and \(\sum_{j\in J}\ell_j\geq P\), we show that \(J\) is a solution to \textsl{K-BPR2}. Firstly, \(\sum_{j\in J}w_j=\sum_{j\in J}b_j\leq B=W\). Also, the profit calculates as \(\sum_{j\in J}p_j=\sum_{j\in J}\ell_j\geq P\).
	
	Given an interdiction strategy \(\gamma\), the corresponding \(p\)-median problem on $G(\gamma)$ can be solved in polynomial time as $G(\gamma)$ still does not contain cycles, cf.~\cite{kariv1979algorithmicmedian}. Thus, the \(p\)-median interdiction problem on trees is contained in \(\NP\) and therefore \(\NP\)-complete.
	
\end{proof}

The \(p\)-median location interdiction problem is \(\NP\)-complete, even on trees in general. In the remainder of the article, we further restrict the problem in order to get a better impression on what makes the problem hard. 

A direct consequence of the \(p\)-median location interdiction problem on trees beeing contained in \(\NP\), is that regarding the interdiction budget as a constant makes the problem polynomial time solvable. This holds true, as having a constant interdiction budget makes the number of possible interdiction strategies polynomial in the instance size. Thus, the procedure of solving a \(p\)-median problem for all interdiction strategies and thereby finding the best strategy runs in polynomial time.
\begin{observation}
	The \(p\)-median interdiction problem on trees is polynomial time solvable if the interdiction budget is considered a constant.
\end{observation}

Further possible restriction are to simplify the graph structure even more, making the edges have unit interdiction costs, making the edges have constant or even unit lengths, or restricting the number $p$ of locations to be placed. Any combinations of these restrictions is also interesting. As analyzing all possible and interesting combinations would exceed the scope of a single article, we focus on the subcase of unit interdiction cost and regard different further restrictions.

In the next section, we consider the \(p\)-median location interdiction problem on paths with unit interdiction costs and afterwards move back to the same problem on trees.

\section{Interdicting a path}\label{sec.path}
In this section, let \(G=(V,E)\) be a graph with \(V=\{v_1,\dots,v_n\}\) and \(E=\{e_i=\{v_i,v_{i+1}\}\colon i=\{1,\dots,n-2\}\). The resulting graph has the structure of a path consisting of \(n\) vertices and \(n-1\) edges. For the remainder, we refer to these types of graphs as paths. \\
In this section we tackle the \(p\)-median location interdiction problem on paths, where we additionally assume unit interdiction costs and \(p=B+1\), i.e.~every component emerging from the interdiction can be equipped with exactly one location. Note that for unit interdiction costs, we may assume \(B\leq n-1\), as a path only contains $n-1$ edges.
We first elaborate on paths with unit lengths. In this case the interdictor can use a simple method to worsen the situation for the locator. This procedure is initially analyzed for one interdiction step (\(B=1\)) and is then generalized to arbitrary interdiction budgets \(B> 1\). After that, we show, how paths with arbitrary lengths can be handled for \(B=1\). 

We want to briefly present the idea of Goldman's algorithm (cf \cite{goldman1971optimal}), since it is needed in the remainder of the article. For a tree \(T\), we start at an arbitrary leaf and compare the weight of that leaf to the total summarized weight of all vertices. As long as the weight of the leaf is less than half of the total weight, we delete the leaf and update the weight of the adjacent vertex by adding the weight of the deleted leaf. In that manner, we iterate over the leaves until we find one with a weight greater or equal to the half of the total weight. This vertex is the \(1\)-median of the original graph. With this method, we are also able to state the \(1\)-median location(s) on a path, which is dependent on the number of the vertices. The optimal solution(s) on a path is to place the new location(s) at vertex \(v_{\nicefrac{n}{2}-1}\) or \(v_{\nicefrac{n}{2}}\) for \(n\) even or at vertex \(v_{\ceil{\nicefrac{n}{2}}}\) for \(n\) odd. 
\subsection{Paths with unit edge weigths}\label{subsec.path.unit}
Let the graph be a path \(P=\left(v_1,e_1,v_2,\dots,e_{n-1},v_n\right)\) with \(\ell\equiv 1\) and \(b\equiv 1\) as described above. We first examine the case for \(B=1\). 
\begin{lemma}\label{lem:b1path}
	Let \(P\) be a path with \(\ell\equiv 1\) and \(b\equiv 1\). The optimal interdiction strategies for the \(p\)-median location interdiction problem are to interdict \(e_1\) or \(e_{n-1}\), yielding in an isolated vertex and a new path of length \(n-2\), i.e. \[\Gamma^*=\{\gamma_1^*=(1,0,\dots,0),\gamma_2^*=(0,\dots,0,1)\},\] where the order of \(\gamma_i^*, i=1,2\) is induced by the order of the edges.
\end{lemma}
\begin{proof}
	As described above, the optimal solution(s) for the \(1\)-median problem are at vertex \(v_{\nicefrac{n}{2}-1}\) or \(v_{\nicefrac{n}{2}}\) for \(n\) even or at vertex \(v_{\ceil{\nicefrac{n}{2}}}\) for \(n\) odd. The resulting objective function value~\opts{1}{P} can then be computed via: 
	\[\opts{1}{P}=\begin{cases}
	\frac{n^2}{4} & n \text{ even}\\
	\frac{n^2-1}{4} & \text{else}
	\end{cases}\]
	Every interdiction strategy with \(B=1\) results in two components of the original path. Since \(p=2\), i.e. one new location in each component, we compute the optimal objective function value for the \(2\)-median problem under the given setting by adding up both objective function values computed separately for every component. Therefore, let \(e_t\) be the interdicted edge for some \(t\in\{1,\ldots,n-1\}\) yielding a separation of the original path \(P\) into \(P_t=(v_1,e_1,v_2,\dots,e_{t-1},v_t)\) and  \(P_{n-t}=(v_{t+1},e_{t+1},\dots,e_{n-1},v_n)\). Then, the objective function for the overall problem of placing one new location in either part is
	\begin{equation}\label{eq.pathp2} z^*=\opts{1}{P_t}+\opts{1}{P_{n-t}}=\frac{1}{4}(2t^2+n^2-2nt-a)\end{equation} with \(a\in\{0,1,2\}\). The goal of the interdictor is to choose \(t\) such that \(z^*\) is maximal. For a given case, i.e. where \(n\) and \(a\) are fixed, we can reduce equation~\ref{eq.pathp2} to the following expression \[t^2-nt=\left(t-\frac{n}{2}\right)^2-\frac{n^2}{4}\] for the computation of the maximum. Since we aim to maximize the latter expression, it can again be reduced to \[\left(t-\frac{n}{2}\right)^2.\]  For \(t\in\{1,\dots,n-1\}\), the maximum is found at \(t_1=1\) or \(t_2=n-1\), which proves the claim.
\end{proof}
This result allows to expand the considerations to \(B\leq n-1\). 
\begin{lemma}
	Let a  path~\(P=\left(v_1,e_1,v_2,\dots,e_{n-1},v_n\right)\) be given with \(\ell\equiv 1\), \(b\equiv 1\). The optimal interdiction strategy under an interdiction budget \(B\leq n-1\) is to successively interdict the edges incident to leaves, thus resulting in \(B\) single vertices and one path of length \(n-1-B\).
\end{lemma}
\begin{proof}
	Consider an optimal interdiction strategy~\(\gamma '\), where at least one of the interdicted edges does not cut off a leaf as depicted in Figure~\ref{fig:pathdetail}. Let \(e_r\) be the described edge.
		\begin{figure}
			\centering
		\begin{tikzpicture}
		[leaf/.style={shape=circle,draw=blue, minimum size=2cm},
		vertex/.style={shape=circle,draw=black, minimum size=.7cm, inner sep=1pt},
		int/.style={color=red, thick}]
			\foreach \p in {-.5,2.5,7,10} \node (\p) at (\p,0) {\(\dots\)};
			\foreach \t/\s in {1/q,4/r,5.5/{r+1},8.5/s} \node[vertex] (\s) at (\t,0) {\tiny\(v_{\s}\)};
		\draw[int, label={u}] (0,0) -- (q);
		\draw[int] (r) -- (r+1);
		\draw[int] (s) -- (9.5,0);
		\draw (q) -- (2,0);
		\draw (r+1) -- (6.5,0);
		\draw (3,0) -- (r);
		\draw (7.5,0) -- (s);
		\end{tikzpicture}
		\caption{Detail of interdicted path \(P\); interdicted edges are depicted in red.}\label{fig:pathdetail}
	\end{figure}
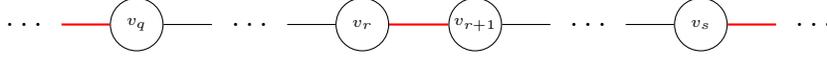
	Given there exists an edge with the stated properties, there also exist sets of vertices \(V_q=\{v_{q}\in V\colon q<r,  \deg(v_q)=1\}\) and \(V_s=\{v_{s}\in V\colon r<s,  \deg(v_s)=1\}\). Let \(v_q\in V_q\) be the vertex with the biggest index and \(v_s\in V_s\) the vertex with the smallest index. This requirement ensures that the component \(P_{qs}= (v_q,e_q,\dots,e_{s-1},v_s)\) of the original path is only interdicted once at \(e_r\).\\
	Now consider \(P_{qs}\), which is interdicted at \(e_r\) with strategy \(\gamma '\), resulting in two new paths. Using the result of Lemma~\ref{lem:b1path}, the optimal objective function value for the interdiction of path \(P_{qs}\) does not decrease by interdicting \(e_q\) instead of \(e_r\). Therefore, successively using this method of shifting the interdicted edges to the leftmost edge of the respective components will also not decrease the overall objective value of \(\gamma '\) yielding an optimal interdiction strategy \(\gamma^*\) as stated. 
\end{proof}
\subsection{Paths with arbitrary lengths}\label{subsec.paths.arbitr}
For the remainder of the section, we assume an arbitrary length function \(\ell\) to be given.\\
One obvious strategy is to interdict all edges successively. In each step at a time, we compute the optimal objective function value for the locator by solving two median location problems on the remaining paths after the current interdiction step. Evaluating over all obtained objective function values yields the best edge to interdict. We present an approach which efficiently iterates over all edges by using an interesting structure of a matrix, which helps computing the locators objective function values. Let a path  \(P=\left(v_1,e_1,v_2,\dots,e_{n-1},v_n\right)\) be given. The median location is found at vertex \(v_{\lceil \nicefrac{n}{2}\rceil} \) for \(n\) odd or at vertex \(v_{\nicefrac{n}{2}}\) or \(v_{\nicefrac{n}{2}+1}\) for \(n\) even. As stated in Section ~\ref{sec:preliminaries}, the objective function value for the \(1-\)median problem on the given path is determined by the total number of times, each edge is crossed to reach all vertices from the median location. Given the structure of a path, these numbers are bounded by \(\lfloor\nicefrac{n}{2}\rfloor\) if \(n\) is odd and by \(\nicefrac{n}{2}\) if \(n\) is even. More precisely, these bounds hold for the edges \(e_{\lfloor\nicefrac{n}{2}\rfloor}\) and \(e_{\lceil\nicefrac{n}{2}\rceil}\) incident to the median location (\(n\) odd) or the edge \(e_{\nicefrac{n}{2}}\) (\(n\) even), respectively. Furthermore, this number decreases by one the closer the edges are to the leaves of the path. Example~\ref{ex:calc.path} shows the case for a path of length \(7\).

\begin{example}\label{ex:calc.path}
	Let \(P=(v_1,e_1,\dots,e_6,v_7)\) a path with length values \(\ell_i, i=\{1,\dots,6\} \) as depicted. With \(n\) odd and the observations above, we get that the median is located at \(v_{\lceil \nicefrac{n}{2}\rceil}=v_4\). Furthermore, we can determine the number of times \(s_i\) the edge \(e_i\) is represented in the objective function value (depicted in green). 
	
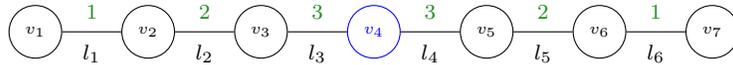
\begin{center}
		\begin{tikzpicture}
			[leaf/.style={shape=circle,draw=blue, minimum size=2cm},
			vertex/.style={shape=circle,draw=black, minimum size=.7cm, inner sep=1pt},
			int/.style={color=red, thick}
			,amount/.style={above, ForestGreen, font=\scriptsize}
			]
				\foreach \t/\s in {1/1,2.5/2,4/3,7/5,8.5/6,10/7} \node[vertex] (\s) at (\t,0) {\tiny\(v_{\s}\)};
				\node[blue, vertex, draw=blue](4) at (5.5,0) {\tiny\(v_4\)};
				\foreach \i/\j in {1/1,2/2,3/3,3/4,2/5,1/6}{
					\pgfmathtruncatemacro\result{\j +1} 
						\draw (\j) -- node[label={[ForestGreen, font=\scriptsize]above:\(\i\)}, label={[font=\scriptsize]below:\(l_\j\)}, inner sep=0pt] {} (\result);
				}
		\end{tikzpicture}
	\end{center}
\end{example}
The information of how often an edge length contributes to the objective function value can be stored in a vector \(S\in\N^{n-1}\), where each entry represents one edge. Multiplying \(S\) with the length vector \(\ell\) yields the optimal objective function value for the path.\\
We use this scheme for the construction of the matrix calculating the objective function values for different interdicted edges. Assume that some edge~\(e_t\), \(t\in \{1,\dots,n-1\}\) gets interdicted. This results in the two paths \(P_{t,1}=(v_1,e_1,v_2,\dots,e_{t-1},v_t)\) and  \(P_{t,2}=(v_{t+1},e_{t+1},\dots,e_{n-1},v_n)\), for which the objective function values can be calculated separately as stated via the vectors \(S_{t,1}=(s_1,\dots,s_{t-1})\in \) for \(P_1\) and \(S_{t,2}=(s_{t+1},\dots,s_{n-1})\) for \(P_2\). The union yields a new vector \(\mathcal{S}_t=(S_{t,1},0,S_{t,2})\). Assuming that we only interdict once, we can proceed to build \(\mathcal{S}_t\) for all edges \(e_t, t=1,\dots,n\) sequentially. The matrix \(\mathscr{S}=(\mathcal{S}_t)_t\in \N^{(n-1)\times(n-1)}\) can again be multiplied with \(\ell\). Evaluating for the biggest objective function value solves the \(1-\)interdiction median problem. Example~\ref{ex:calc.pathcomplete} shows the matrix for the path of Example~\ref{ex:calc.path}.
\begin{example}\label{ex:calc.pathcomplete}
	Let \(P\) be the path of Example~\ref{ex:calc.path}. Furthermore, let the length vector \(\ell\) be given. The matrix obtained via the presented method is as follows:	
	\renewcommand*{\arraystretch}{.5}
	\[\begin{pmatrix}
		0 & 1 & 2 & 3 & 2 & 1\\
		1 & 0 & 1 & 2 & 2 & 1\\
		1 & 1 & 0 & 1 & 2 & 1\\ 
		1 & 2 & 1 & 0 & 1 & 1\\
		1 & 2 & 2 & 1 & 0 & 1\\
		1 & 2 & 3 & 2 & 1 & 0
	\end{pmatrix}\]
\end{example}
As stated, the matrix \(\mathscr{S}\) is of size \((n-1)\times (n-1)\), where each row \(i\in\{1,\dots, n-1\}\) represents the objective function value if edge \(e_i\) gets interdicted. Therefore, the procedure explained runs in time \(\Ol(n^2)\) since the multiplication of matrix \(\mathscr{S}\) with the given length vector \(\ell\) is of stated time. We want to mention, that an increase of interdiction steps adds factor \(n\) to the running time for each single interdiction. This is due to the fact, that we need to consider every possible combination for the edges, that get interdicted.
\section{Interdicting a tree}\label{sec.tree}
In this section, we show, how to interdict a tree \(T=(V,E)\) with \(\ell\equiv 1\), \(b\equiv 1\), \(B=1\) and \(p=B+1\).

\begin{theorem}\label{thm:interdictionLeafOnTree}
	For \(T=(V,E)\) with \(\ell\equiv 1\), \(b\equiv 1\), \(B=1\) and \(p=B+1\), the optimal interdiction strategy is as follows: among all leaves of \(T\), find one, that has the least distance to at least one optimal solution of the  \(1-\)median location problem on \(T\). Interdict the edge incident to this leaf.
\end{theorem}
\begin{proof}
	Since the proof is constructive, consider the exemplary tree in Figure~\ref{fig:lemtree} which illustrates the used structures.
	
	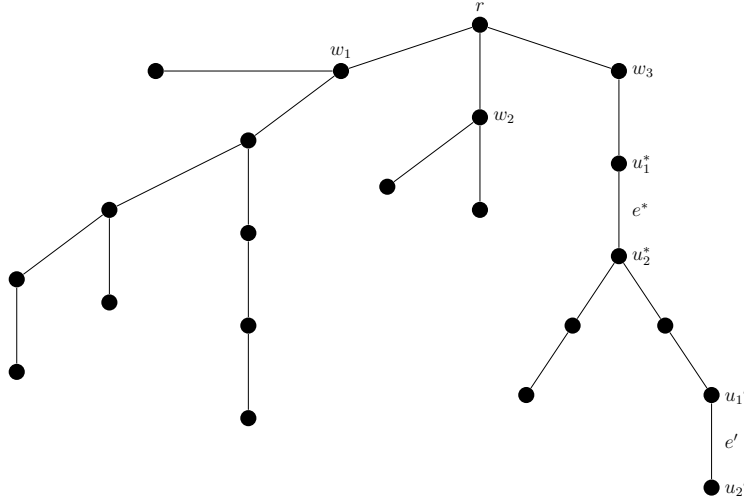
\begin{figure}
		\centering
		\resizebox{10cm}{!}{%
		\begin{tikzpicture}
		\begin{scope}
		[every node/.style={shape=circle,fill=black, inner sep=3.5}]
			\foreach \x/\y/\z in {0/2.5/1,
										0/4.5/2,
										2/4/3,
										2/6/4,
										3/9/5,
										5/1.5/6,
										5/3.5/7,
										5/5.5/8,
										5/7.5/9,
										7/9/10,
										8/6.5/11,
										10/6/12,
										10/8/13,
										10/10/14,
										11/2/15,
										12/3.5/16,
										13/5/17,
										13/7/18,
										13/9/19,
										14/3.5/20,
										15/0/21,
										15/2/22} \node (\z) at (\x,\y) {};
			\end{scope}
			\node [label={[label distance=0cm]90:\(r\)}] at (14){};
			\node [label={[label distance=0cm]0:\(u_1^*\)}] at (18){};
			\node [label={[label distance=0cm]0:\(u_2^*\)}] at (17){};
			\node [label={[label distance=0cm]0:\({u_1}'\)}] at (22){};
			\node [label={[label distance=0cm]0:\({u_2}'\)}] at (21){};
			\node [label={[label distance=0cm]90:\(w_1\)}] at (10){};
			\node [label={[label distance=0cm]0:\(w_2\)}] at (13){};
			\node [label={[label distance=0cm]0:\(w_3\)}] at (19){};
			
			\node [label={[label distance=0cm]0:\(e^*\)}] at (13,6){};
			\node [label={[label distance=0cm]0:\(e'\)}] at (15,1){};
			
			\draw (1) -- (2) -- (4) -- (3);
			\draw (9) -- (10) -- (5);
			\draw (14) -- (10);
			\draw (4) -- (9) -- (8) -- (7) -- (6);
			\draw (11) -- (13) -- (12) ;
			\draw (15) -- (16) -- (17) ;
			\draw (13) -- (14) -- (19) -- (18) -- (17) -- (20) -- (22) -- (21) ;
		\end{tikzpicture}
	}
		\caption{Exemplary tree for Theorem~\ref{thm:interdictionLeafOnTree}.}\label{fig:lemtree}
	\end{figure}
	Let \(r\in V(T)\) be an optimal solution to the \(1-\)median location problem on \(T\). Also, let \(T\) be rooted in \(r\). Let \(f=(v_1,v_2)\in E(T)\) be an edge with \(v_2\) being a leaf that fulfills \begin{align}\label{lem.mindist}d(r,v_2)=\min_{\substack{r^*\in X^*\\ l \text{ leaf}}} d(r^*,l) .\end{align}
	For all neighbors \(w\in N(r)\), it holds that \begin{align}\label{ineq.neighbor}\left|T_{w}\right|\leq\left|T-T_{w}\right|. \end{align} 
	Assume, that \(\card{T_w}>\card{T-T_w}\). Then, \(w\) would yield a better objective function value for the \(1-\)median location problem than \(r\), which is a contradiction to the assumption, that \(r\) is optimal for said problem.	
	Now, for every edge \(e=(u_1,u_2)\in E(T)\) with \(d(r,u_1)<d(r,u_2)\), it is 
	\begin{align}\label{lem:ineq} \opts{2}{T-e}\leq \opts{1}{T}-(d(r,u_2)\cdot \card{T_{u_2}} ). \end{align}
	This is, since placing a second location in the tree \(T_{u_2}\) saves at least \(\card{T_{u_2}}\) times the distance \(d(r,u_2)\) in comparison to the solution of the \(1-\)median problem on the original tree \(T\).

	We aim at finding \(e^*=(u_1^*,u_2^*)\) such that \(d(r,u_2^*)\cdot \card{T_{u_2^*}}=\min_{e\in E(T)}d(r,u_2)\cdot \card{T_{u_2}}\). This leads to the biggest right hand side of inequality~\eqref{lem:ineq}, thus providing an upper bound for the optimal objective function value \(\opts{2}{T-e}\), which the interdictor maximizes. We claim, that \(u_2^*\) is a leaf. 
	
	\begin{description}
	\item[Case 1:] \(d(r,u_2^*)>1\).
	Suppose, \(u_2^*\)  is not a leaf and let instead \(e'=(u_1',u_2')\in T_{u_2^*}\) with \(u_2'\) being a leaf. Then, \(\card{T_{u_2^*}}\geq d(r,u_2')-d(r,u_2^*)+1\). Using this, we get 
	\begin{align}
		d(r,u_2^*)\cdot\card{T_{u_2^*}}&\geq d(r,u_2^*) \cdot d(r,u_2')-d^2(r,u_2^*)+d(r,u_2^*) \nonumber\\
		&= d(r,u_2')\left(d(r,u_2^*)-\frac{d(r,u_2^*)(d(r,u_2^*)-1)}{d(r,u_2')} \right)\nonumber\\ \label{lem:leaf}
		&> d(r,u_2')\left(d(r,u_2^*)-(d(r,u_2^*)-1) \right) \\ 
		&= d(r,u_2') = d(r,u_2')\cdot\card{T_{u_2'}}\nonumber
	\end{align}
	where \eqref{lem:leaf} follows from the fact that \(\frac{d(r,u_2^*)}{d(r,u_2')}<1\) under the assumptions made. This is a contradiction to  \(d(r,u_2^*)\cdot \card{T_{u_2^*}}=\min_{e\in E(T)}d(r,u_2)\cdot \card{T_{u_2}}\).	
		\item[Case 2:] \(d(r,u_2^*)=1\). In this case, \eqref{lem:leaf} does not hold. We need to distinguish between two cases.
		\begin{description}
			\item[Case 2.1:] \(u_2^*\) is a leaf. \(\checkmark\)
			\item[Case 2.2:] Assume, that \(u_2^*\) is not a leaf. Again, we distinguish between two cases.
			\begin{description}
				\item[Case 2.2.1:] Assume, that \(T_{u_2^*}\) is not a path. Let \(u_2' \in T_{u_2^*}\) be a leaf. There is at least one vertex \(u''\in T_{u_2^*}\) for which \(d(u'')\geq 2\). Thus, it holds that 
				\begin{align*}
					d(u_2^*, u_2')&\leq \card{T_{u_2^*}}-2\\
					\iff d(r, u_2')-1 & \leq \card{T_{u_2^*}}-2\\
					\iff d(r, u_2') & < \card{T_{u_2^*}}
				\end{align*}
				Since \(d(r,u_2^*)=\card{T_{u'}}=1\), we get \(d(r, u') \cdot \card{T_{u'}} < d(r,u_2^*) \cdot {T_{u_2^*}}\), which is a contradiction to \(d(r,u_2^*)\cdot \card{T_{u_2^*}}=\min_{e\in E(T)}d(r,u_2)\cdot \card{T_{u_2}}\).	
				\item[Case 2.2.2:] Assume \(T_{u_2^*}\) is a path. Then, we find that \\ \(\min_{e\in E(T)}d(r,u_2)\cdot \card{T_{u_2}}=d(r,u_2^*)\cdot \card{T_{u_2^*}}=d(r,u_2')\cdot \card{T_{u_2'}}\) with \(u_2'\in T_{u_2^*}\) being a leaf.
			\end{description}
		\end{description}
	\end{description}
	Coming back to inequality~\eqref{lem:ineq} we conclude, that the right hand side is biggest, if for \(e^*=(u_1^*,u_2^*)\), vertex \(u_2^*\) is a leaf. \\
	We now want to use the term on the right side of inequality~\eqref{lem:ineq} as an upper bound for the objective function value of the interdiction problem. Therefore, let \(f=(v_1,v_2)\) be the edge from Assumption~\eqref{lem.mindist}, which gets interdicted. Then, we find the following.\\
	If the optimal solution \(r\) to the \(1-\)median problem on \(T\) is part of the optimal solution of the \(2-\)median problem on \((T-f)\) -- together with the leaf \(v_2\) -- then it holds that
	\[\opts{2}{T-f}=\opts{1}{T}-d(r,v_2).\]
	 Now suppose, that the optimal solution \(r\) to the \(1-\)median problem on \(T\) is not part of the optimal solution of the \(2-\)median problem on \((T-f)\) anymore. That means there exists \( s\in N(r)\) in the neighborhood of \(r\) which is part of the optimal solution together with leaf \(v_2\). Using the fact, that \(d(r,s)=1\), we get that
	\begin{align}\label{lem.bound}
		\opts{2}{T-f}= \opts{1}{T} - d(r,v_2)-1.
	\end{align}
 Note, that \(\opts{2}{T-f}\) does not depend on the specific choice of \(f\), meaning the objective function value is the same for all \(f\) fulfilling \eqref{lem.mindist}. In any case it holds true that 	\begin{align}\label{lem.bound2}
 \opts{2}{T-f} \geq \opts{1}{T} - d(r,v_2)-1.
 \end{align}
 Now we distinguish between two cases.
\begin{description}
	\item[Case 1.] Let \(f'=(v'_1,v'_2)\) such that \(v'_2\) is a leaf, but not fulfilling assumption~\eqref{lem.mindist}. Therefore, \(d(r,v_2')>d(r,v_2)\). Then we get that 
	\begin{align}
	\label{lem.case11}	\opts{2}{T-f'}&\leq \opts{1}{T}-d(r,v'_2)\\
	\label{lem.case12}	&\leq \opts{1}{T}-d(r,v_2)-1\\ 
	\label{lem.case13}	&\leq \opts{2}{T-f}
	\end{align}
	where~\eqref{lem.case11} follows from inequality~\eqref{lem:ineq}, \eqref{lem.case12} follows from the estimation of the distances in the current case, and \eqref{lem.case13} is due to inequality~\eqref{lem.bound2}. 
	\item[Case 2.] Let \(f'=(v'_1,v'_2)\) such that \(v'_2\) is not a leaf. Then, with \(q_2\in T_{v'_2}\) being a leaf, it is
	\begin{align}
	\label{lem.case21}	\opts{2}{T-f'}&\leq \opts{1}{T}-d(r,v'_2)\cdot \card{T_{v'_2}}\\
	\label{lem.case22}	&\leq \opts{1}{T}-d(r,q_2)-1 \\ 
	\label{lem.case23}	&\leq \opts{1}{T}-d(r,v_2)-1 \\
	\label{lem.case24}	&\leq \opts{2}{T-f}.
	\end{align}
	The first line~\eqref{lem.case21} again follows from inequality~\eqref{lem:ineq}, \eqref{lem.case22} is due to the definition of \(q_2\), \eqref{lem.case23} follows from condition~\eqref{lem.mindist} and the last estimation~\eqref{lem.case24} follows from inequality~\eqref{lem.bound2}.
\end{description} 
	Combining our results of cases 1 and 2, we find that choosing an edge \(f\) for the interdiction step, which fulfills condition~\eqref{lem.mindist} yields the optimal objective function value for the interdictor, namely the biggest objective function value \(\opts{2}{T-f}\).
	\end{proof}

\begin{observation}
	In case 2.2.2, where \(d(r,u_2^*)=1\) and \(T_{u_2^*}\) is a path, we have seen that \(\min_{e\in E(T)}d(r,u_2)\cdot \card{T_{u_2}}=d(r,u_2^*)\cdot \card{T_{u_2^*}}=d(r,u_2')\cdot \card{T_{u_2'}}\) for the leaf \(u_2'\). We stated, that we choose the edge connecting the leaf for our interdiction strategy. In fact, this is mandatory in every case, except for \(u_2'\) being the successor of \(u_2^*\). Let  \(\card{T_{u_2^*}}>2\). Then, if the edge \(e^*\) is interdicted, the locator can improve the objective function value of the remainder of the path, which is not connected to the main tree containing \(r\) anymore, with the placement of the new location. 
\end{observation}

\begin{remark}
The procedure described above needs to compute an optimal solution of the 1-median location problem and the distance from all leaves to this solution. The former computation can be done in time~$\mathcal{O}(n)$ (cf. \cite{goldman1971optimal}), whereas the latter can be done in time~$\mathcal{O}(n^2)$ by a breadth-first-search on every vertex. Thus, in total, the procedure runs in $\mathcal{O}(n^2)$.
\end{remark}

\begin{remark}
	It is not possible to apply the procedure to generalized problems. There are several alterations possible, where one can change the input in one parameter while the rest of the problem setup stays the same. This includes a change of the interdiction budget (\(B>1\)) or an arbitrary length function on the edges. 
	For the first problem, we immediately see, that simply choosing the \(B\)  leaves closest to the median location and interdict them in one step, can lead to the following problems. First of all, it is not clear, whether \(B\) leaves exist in the original graph. But even if there are, the procedure does not necessarily  give the optimal interdiction strategy as example~\ref{ex:onestepno} shows.
	\begin{example}\label{ex:onestepno}
		Let a tree \(T=(V,E)\) be given as in Figure~\ref{fig:onestepnotree}.
		\begin{figure}[h!]
			\centering
			\begin{tikzpicture}
				[vertex/.style={shape=circle,fill=black, inner sep=2},
				amount/.style={pos=.5, font=\scriptsize}, inner sep=.5]
				\foreach \x/\y/\z in {
					-2.75/.75/1,
					-1.75/.75/2,
					-1.75/-.75/3,
					-1/0/4,
					0/0/5,
					2/0/7,
					3.75/.75/8,
					2.75/.75/9,
					3.75/-.75/10,
					2.75/-.75/11} \node[vertex] (\z) at (\x,\y) {};
				\node[shape=circle,fill=blue, inner sep=2] (6) at (1,0) {};
				\draw (1) --  (2) -- (4) -- (5) --  (6) --  (7) -- (9) -- (8) ;
				\draw (3) -- (4);
				\draw (7) -- (11) -- (10);
			\end{tikzpicture}
			\caption{Exemplary tree, optimal median location depicted in blue.}\label{fig:onestepnotree}
		\end{figure}
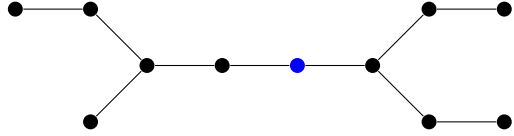
		Let the interdiction budget be given with \(B=3\). Then, if we choose the \(3\) leaves closest to the optimal location, we yield the graph in Figure~\ref{fig:onestepnotree1}. The optimal objective function value of the locator is \(15\) in this case, but there are better interdiction strategies. One of the optimal strategies is depicted in Figure~\ref{fig:onestepnotree2} and leaves the locator with an optimal objective function value of \(16\). 
		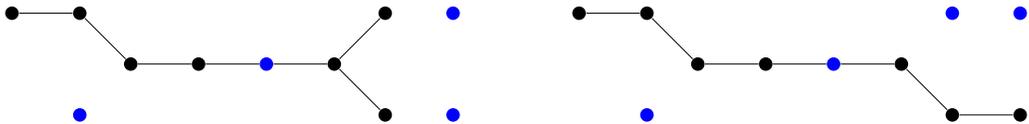
\begin{figure}[h!]
			\centering
			\begin{subfigure}[t]{.45\textwidth}
				\centering
				\resizebox{\textwidth}{!}{%
			\begin{tikzpicture}
				[vertex/.style={shape=circle,fill=black, inner sep=2},
				amount/.style={pos=.5, font=\scriptsize}, inner sep=.5]
				\foreach \x/\y/\z in {
					-2.75/.75/1,
					-1.75/.75/2,
					-1/0/4,
					0/0/5,
					2/0/7,
					2.75/.75/9,
					2.75/-.75/11} \node[vertex] (\z) at (\x,\y) {};
				\node[shape=circle,fill=blue, inner sep=2] (6) at (1,0) {};
				\node[shape=circle,fill=blue, inner sep=2] (3) at (-1.75,-.75) {};
				\node[shape=circle,fill=blue, inner sep=2] (8) at (3.75,.75) {};
				\node[shape=circle,fill=blue, inner sep=2] (10) at (3.75,-.75) {};
				\draw (1) --  (2) -- (4) -- (5) --  (6) --  (7) -- (9) ;
				\draw (7) -- (11);
			\end{tikzpicture}
		}
				\caption{Exemplary tree after interdiction step following the idea of Theorem~\ref{thm:interdictionLeafOnTree}, median locations depicted in blue.}\label{fig:onestepnotree1}
			\end{subfigure}
			\hfill
			\begin{subfigure}[t]{.45\textwidth}
				\centering
				\resizebox{\textwidth}{!}{%
			\begin{tikzpicture}
				[vertex/.style={shape=circle,fill=black, inner sep=2},
				amount/.style={pos=.5, font=\scriptsize}, inner sep=.5]
				\foreach \x/\y/\z in {
					-2.75/.75/1,
					-1.75/.75/2,
					-1/0/4,
					0/0/5,
					2/0/7,
					3.75/-.75/10,
					2.75/-.75/11} \node[vertex] (\z) at (\x,\y) {};
				\node[shape=circle,fill=blue, inner sep=2] (6) at (1,0) {};
				\node[shape=circle,fill=blue, inner sep=2] (3) at (-1.75,-.75) {};
				\node[shape=circle,fill=blue, inner sep=2] (8) at (3.75,.75) {};
				\node[shape=circle,fill=blue, inner sep=2] (9) at (2.75,.75) {};
				\draw (1) --  (2) -- (4) -- (5) --  (6) --  (7);
				\draw (7) -- (11) -- (10);
			\end{tikzpicture}
		}
				\caption{Exemplary tree after optimal interdiction step, median locations depicted in blue.}\label{fig:onestepnotree2}
			\end{subfigure}
		\caption{Exemplary tree of Figure~\ref{fig:onestepnotree} after different interdiction strategies}
		\end{figure}
	\end{example} 	 
	For the case of arbitrary lengths on the edges, we find, that in case 1 of the proof, the inequality  \(\card{T_{u_2^*}}\geq d(r,u_2')-d(r,u_2^*)+1\) relies on the fact, that the distances can be calculated via the amount of vertices, which is only possible because of the unit length values of the edges. In fact, consider the following minimal example~\ref{ex:arbitraryno} illustrating that the proposed method does not work for arbitrary lengths.
	\begin{example}\label{ex:arbitraryno}
		Let a tree \(T=(V,E)\) be given as in Figure~\ref{fig:arbitrarynotree} and \(B=1\).
			\begin{figure}[h!]
				\centering
				\begin{tikzpicture}
					[vertex/.style={shape=circle,fill=black, inner sep=2},
					amount/.style={pos=.5, font=\scriptsize}, inner sep=.5]
					\foreach \x/\y/\z in {-1.75/.75/1,
						-1.75/-.75/2,
						-1/0/3,
						1/0/5,
						1.75/-.75/6,
						1.75/.75/7} \node[vertex] (\z) at (\x,\y) {};
					\node[shape=circle,fill=blue, inner sep=2] (4) at (0,0) {};
					\draw(1) -- node[amount, above right] {\(10\)} (3) -- node[amount, above]{\(1\)} (4) -- node[amount, above]{\(1\)} (5) -- node[amount, below left]{\(10\)} (6) ;
					\draw (2) -- node[amount, below right]{\(10\)} (3);
					\draw (5) -- node[amount, above left]{\(10\)}  (7);
				\end{tikzpicture}
			\caption{Exemplary tree, optimal median location depicted in blue.}\label{fig:arbitrarynotree}
			\end{figure}
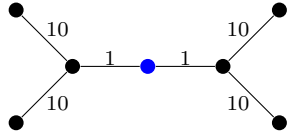
		Following the procedure proposed in Theorem~\ref{thm:interdictionLeafOnTree}, one of the four edges incident to the leaves gets interdicted resulting in the tree in Figure~\ref{fig:arbitrarynotree1}. The optimal objective function value of the locator is \(35\), while the optimal interdiction strategy depicted in Figure~\ref{fig:arbitrarynotree2} leaves the locator with an optimal objective function value of \(41\). 
			\begin{figure}[h!]
				\centering
				\begin{subfigure}{.45\textwidth}
					\centering
				\begin{tikzpicture}
					[vertex/.style={shape=circle,fill=black, inner sep=2},
					amount/.style={pos=.5, font=\scriptsize}, inner sep=.5]
					\foreach \x/\y/\z in {-1.75/.75/1,
						-1.75/-.75/2,
						-1/0/3,
						1/0/5,
						1.75/-.75/6} \node[vertex] (\z) at (\x,\y) {};
					\node[shape=circle,fill=blue, inner sep=2] (4) at (0,0) {};
					\node[shape=circle,fill=blue, inner sep=2] (7) at (1.75,.75) {};
					\draw(1) -- node[amount, above right] {\(10\)} (3) -- node[amount, above]{\(1\)} (4) -- node[amount, above]{\(1\)} (5) -- node[amount, below left]{\(10\)} (6) ;
					\draw (2) -- node[amount, below right]{\(10\)} (3);
				\end{tikzpicture}
				\caption{Exemplary tree after interdiction step following Theorem~\ref{thm:interdictionLeafOnTree}, median locations depicted in blue.}\label{fig:arbitrarynotree1}
			\end{subfigure}
		\hfill
			\begin{subfigure}{.45\textwidth}
				\centering
			\begin{tikzpicture}
				[vertex/.style={shape=circle,fill=black, inner sep=2},
				amount/.style={pos=.5, font=\scriptsize}, inner sep=.5]
				\foreach \x/\y/\z in {-1.75/.75/1,
					-1.75/-.75/2,
					0/0/4,
					1.75/-.75/6,
					1.75/.75/7} \node[vertex] (\z) at (\x,\y) {};
				\node[shape=circle,fill=blue, inner sep=2] (3) at (-1,0) {};
				\node[shape=circle,fill=blue, inner sep=2] (5) at (1,0) {};
				\draw (1) -- node[amount, above right] {\(10\)} (3) -- node[amount, above]{\(1\)} (4);
				\draw (5) -- node[amount, below left]{\(10\)} (6) ;
				\draw (2) -- node[amount, below right]{\(10\)} (3);
				\draw (5) -- node[amount, above left]{\(10\)}  (7);
			\end{tikzpicture}
			\caption{Exemplary tree after optimal interdiction step, median locations depicted in blue.}\label{fig:arbitrarynotree2}
		\end{subfigure}
	\caption{Exemplary tree of Figure~\ref{fig:arbitrarynotree} after different interdiction strategies}
		\end{figure}
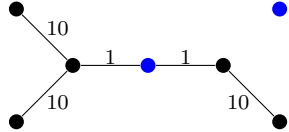
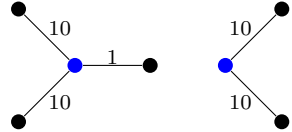
	\end{example} 	 
\end{remark}

\section{Conclusion}\label{sec:conclusion}
In this article, we introduced the \(p\)-median location interdiction problem (MLIP). We proved MLIP to be \(\NP\)-hard on trees in Theorem~\ref{thm:compl}. We then considered the MLIP where the underlying graph is given by a path. For the case of unit lengths, we proved that interdicting the edge incident to a leaf is an optimal interdiction strategy. This strategy can be applied iteratively, if more than one edge can get interdicted. For the case of arbitrary lengths, we showed that an optimal interdiction strategy can be computed in polynomial time. Furthermore, we proposed a polynomial time algorithm for the case of unit lengths on a tree. 

The cases of arbitrary lengths on a tree graph as well as multiple interdiction of edges may be considered in the future. It would also be interesting to investigate the MLIP on different graph classes, e.g. on series-parallel graphs. Moreover, other types of location problems may be investigated in the context of interdiction.

\section*{Acknowledgments}
This work is partially supported by the European Social Fund+ (ESF+), project SchuMaMoMINT+, and by project Ageing Smart funded by Carl-Zeiss-Stiftung.

\bibliography{mybibfile.bib}

\end{document}